\documentclass[10pt,conference]{IEEEtran}

\usepackage[linesnumbered,ruled]{algorithm2e}
\usepackage{algorithmic}

\usepackage{amsmath, amssymb, amsfonts}
\usepackage{colortbl, graphicx, mathrsfs}
\usepackage{enumerate}
\usepackage{pstricks, pst-node}

\newtheorem{definition}{Definition}
\newtheorem{theorem}{Theorem}
\newtheorem{lemma}{Lemma}

\begin{document}

\title{Canonical Coin Systems for Change-Making Problems}

\author{\IEEEauthorblockN{Xuan Cai}
\IEEEauthorblockA{Department of Computer Science and Engineering,
Shanghai Jiao Tong University\\
Shanghai 200240, China\\
Email: caixuanfire@sjtu.edu.cn}}

\maketitle

\begin{abstract}
The Change-Making Problem is to represent a given value with the
fewest coins under a given coin system. As a variation of the
knapsack problem, it is known to be NP-hard. Nevertheless, in most
real money systems, the greedy algorithm yields optimal solutions.
In this paper, we study what type of coin systems that guarantee the
optimality of the greedy algorithm. We provide new proofs for a
sufficient and necessary condition for the so-called
\emph{canonical} coin systems with four or five types of coins, and
a sufficient condition for non-canonical coin systems, respectively.
Moreover, we present an $O(m^2)$ algorithm that decides whether a
tight coin system is canonical.
\end{abstract}


\IEEEpeerreviewmaketitle

\section{Introduction}
The Change-Making Problem comes from the following scenario: in a
shopping mall, the cashier needs to make change for many values of
money based on some \emph{coin system}\footnote[2]{\scriptsize{In
this paper, all the variables range over the set $\mathbb{N}$ of
natural numbers.}} $\$=\langle c_1,c_2,\cdots,c_m\rangle$ with
$1=c_1<c_2<\cdots<c_m$, where $c_i$ denotes the value of the $i$-th
type of coin in $\$$. For example, the cent, nickel, dime and
quarter are four types of US coins, and the corresponding coin
system is $\$=\langle1,5,10,25\rangle$. Since the reserved coins are
limited in reality, the cashier has to handle every exchange with as
few coins as possible.

Formally, the Change-Making Problem is to solve the following
integer programming problem with respect to a given value $x$.
\[
\begin{array}{rl}
\min&\sum_{i=1}^m\alpha_i\medskip\\
\text{s.t.}&\sum_{i=1}^mc_i\alpha_i=x,\quad c_i\geq0\\
\end{array}
\]

As usual, we call a feasible solution
$(\alpha_1,\alpha_2,\cdots,\alpha_m)$ of the above integer
programming problem a \emph{representation} of $x$ under $\$$. If
this representation satisfies $\sum_{j=1}^{i-1}\alpha_j c_j<c_i$ for
$2\leq i\leq m$, then it is the \emph{greedy representation} of $x$,
denoted by $\mathrm{GRD}_{\$}(x)$, and
$|\mathrm{GRD}_{\$}(x)|=\sum_{i=1}^m\alpha_i$ is its size.
Similarly, we also call the optimal solution
$(\beta_1,\beta_2,\cdots,\beta_m)$ the \emph{optimal representation}
of $x$, denoted by $\mathrm{OPT}_{\$}(x)$, and
$|\mathrm{OPT}_{\$}(x)|=\sum_{i=1}^m\beta_i$ is its size.\medskip

\subsection{Problem Statement}
The Change-Making Problem is NP-hard
\cite{lue}\cite{gar:joh}\cite{mag:nem:tro} by a polynomial reduction
from the knapsack problem. There are a large number of
pseudo-polynomial exact algorithms \cite{kol}\cite{nem:ull} solving
this problem, including the one using dynamic programming
\cite{wri}. However, the greedy algorithm, as a simpler approach,
can produce the optimal solutions for many practical instances,
especially canonical coin systems.

\begin{definition}
A coin system $\$$ is canonical if
$|\mathrm{GRD}_{\$}(x)|=|\mathrm{OPT}_{\$}(x)|$ for all $x$.
\end{definition}
For example, the coin system $\$=\langle1,5,10,25\rangle$ is
canonical. Accordingly, the cashier can easily create the optimal
solution by repeatedly taking the largest coin whose value is no
larger than the remaining amount.

\begin{definition}
A coin system $\$$ is non-canonical if there is an $x$ with
$|\mathrm{GRD}_{\$}(x)|>|\mathrm{OPT}_{\$}(x)|$, and such $x$ is
called a counterexample of $\$$.
\end{definition}

\begin{definition}
A coin system $\$=\langle1,c_2,\cdots,c_m\rangle$ is tight if it has
no counterexample smaller than $c_m$.
\end{definition}

For example, both $\$_1=\langle1,7,10,11\rangle$ and
$\$_2=\langle1,7,10,50\rangle$ are non-canonical, and $\$_1$ is
tight but $\$_2$ is not. It is easy to verify that $14$ is the
counterexample for them, i.e.,
$\mathrm{GRD}_{\$_1}(14)=\mathrm{GRD}_{\$_2}(14)=(3,0,0,1)$ and
$\mathrm{OPT}_{\$_1}(14)=\mathrm{OPT}_{\$_2}(14)=(0,2,0,0)$.

Nowadays, canonical coin systems have found numerous applications in
many fields, e.g., finance \cite{nem:ull}, management \cite{eil} and
computer networks \cite{ger:kle}. It is desirable to give a full
characterization of them.

\subsection{Related Work}
Chang and Gill \cite{cha:gil} were the first to study canonical coin
systems. They showed that there must be a counterexample $x$ of the
non-canonical $\$=\langle1,c_2,\cdots,c_m\rangle$ such that $c_3\leq
x<\frac{c_m(c_mc_{m-1}+c_m-3c_{m-1})}{c_m-c_{m-1}}$.

Concerning the smallest counterexamples of noncanonical coin
systems, Tien and Hu established the following two important results
in \cite{tie:hu}.

\begin{theorem}\label{dis}
Let $x$ be the smallest counterexample of the non-canonical coin
system $\$=\langle1,c_2,\cdots,c_m\rangle$. Then $\alpha_i\beta_i=0$
for all $i\in[1,m]$ such that
$\mathrm{OPT}_{\$}(x)=(\alpha_1,\alpha_2,\cdots,\alpha_m)$ and
$\mathrm{GRD}_{\$}(x)=(\beta_1,\beta_2,\cdots,\beta_m)$.
\end{theorem}

\begin{theorem}\label{th}
Let $\$_1=\langle1,c_2,\cdots,c_m\rangle$ and
$\$_2=\langle1,c_2,\cdots,c_m,c_{m+1}\rangle$ be two coin systems
such that $\$_1$ is canonical but $\$_2$ is not. Then there is some
$k$ such that $k\cdot c_m<c_{m+1}<(k+1)\cdot c_m$ and $(k+1)\cdot
c_m$ is a counterexample of $\$_2$.
\end{theorem}

These two results not only imply that many coin systems are
canonical such as positive integer arithmetic progressions,
geometric progressions and the Fibonacci sequence but also are the
starting point of a lot of subsequent work.

Building on Theorem~\ref{dis}, Kozen and Zaks \cite{koz:zak} gave a
tight range of smallest counterexamples of non-canonical coin
systems:

\begin{theorem}\label{tkz1}
Let $\$=\langle1,c_2,\cdots,c_m\rangle$ be a coin system. If $\$$ is
not canonical, then the smallest counterexample lies in the range
\[
c_3+1<x<c_{m-1}+c_m.
\]
Furthermore, these bounds are tight.
\end{theorem}

Moreover, they gave a necessary and sufficient condition of the
canonical coin system with three types of coins in \cite{koz:zak}.

\begin{theorem}\label{tkz2}
The coin system $\$=\langle1,c_2,c_3\rangle$ is non-canonical if and
only if $0<r<c_2-q$ where $c_3=qc_2+r$ and $r\in[0,c_2-1]$.
\end{theorem}

Here, we provide a proof of this theorem which will be used later.

\begin{proof}[Proof of Theorem~\ref{tkz2}]
\begin{enumerate}
\item[($\Leftarrow$)]Consider
the integer $x=c_2+c_3-1$, $\mathrm{GRD}_{\$}(x)=(c_2-1,0,1)$. Since
$(r-1,q+1,0)$ is a representation of $x$, we have
$|\mathrm{OPT}_{\$}(x)|\leq r+q$. By the precondition $0<r<c_2-q$,
it is easy to see $|\mathrm{OPT}_{\$}(x)|<|\mathrm{GRD}_{\$}(x)|$.
Thus, $c_3+c_2-1$ is a counterexample of $\$$ and it is
non-canonical.\smallskip

\item[($\Rightarrow$)]Since $\$$ is non-canonical, let $x$ be the
smallest counterexample. By Theorem~\ref{tkz1},
$x\in[c_3+2,c_2+c_3-1]$. Without loss of generality, let
$\mathrm{GRD}_{\$}(x)=(e,0,1)$ and $\mathrm{OPT}_{\$}(x)=(0,k,0)$
with $e\in[1,c_2-1]$. Then we have $x=c_3+e=kc_2$, i.e., $q=k-1<e$
and $r=c_2-e\geq1$. Thus, $0<r<c_2-q$.
\end{enumerate}
\end{proof}

Pearson \cite{pea} proved the following theorem that characterizes
the smallest counterexample of the non-canonical coin system.

\begin{theorem}\label{tp}
Let $x$ be the smallest counterexample of the non-canonical coin
system $\$=\langle1,c_2,\cdots,c_m\rangle$. If
$\mathrm{OPT}_{\$}(x)=(\underbrace{0,\cdots,0}_0,\beta_l,\cdots,\beta_r,\underbrace{0,\cdots,0}_0)$
with $\beta_l,\beta_r>0$, then
$\mathrm{GRD}_{\$}(c_{r+1}-1)=(\alpha_1,\cdots,\alpha_{l-1},\beta_l-1,\beta_{l+1},\cdots,\beta_r,\underbrace{0,\cdots,0}_0)$.
\end{theorem}

Based on this theorem, he gave an $O(m^3)$ algorithm to decide
whether a coin system $\$=\langle1,c_2,\cdots,c_m\rangle$ is
canonical.

Recently, Niewiarowska and Adamaszek \cite{ada:nie} investigate the
structure of canonical coin systems and present a series of
sufficient conditions of non-canonical coin systems.

\subsection{Our Contribution}
In this paper, we study canonical coin systems for the Change-Making
Problem and obtain the following results.

\begin{itemize}
\item We give an easy proof for a sufficient and necessary
condition of canonical coin systems with four or five types of
coins.

\item We provide a new proof for natural sufficient condition of non-canonical coin
systems.

\item We present an $O(m^2)$ algorithm that decides whether a tight coin
system is canonical.
\end{itemize}

The rest of this paper is organized as follows. In Section 2, we
study canonical coin systems with four types of coins. In Section 3,
we extend the study to canonical coin systems with five types of
coins. Section 4 introduces tight canonical coin systems and Section
5 presents an $O(m^2)$ algorithm that decides whether a tight coin
system is canonical. Finally, in Section 6, we address some of the
questions left open and discuss future work.

\section{Coin System with Four Types of Coins}
In this section, we study canonical coin systems with four types of
coins and give a full characterization of them based on
Theorem~\ref{tkz2}.

\begin{theorem}\label{four}
A coin system $\$=\langle1,c_2,c_3,c_4\rangle$ is non-canonical if
and only if $\$$ satisfies exactly one of the following conditions:

\begin{enumerate}
\item $\langle1,c_2,c_3\rangle$ is non-canonical.

\item $|\mathrm{GRD}_{\$}((k+1)\cdot c_3)|>k+1$ with $k\cdot c_3<c_4<(k+1)\cdot c_3$.
\end{enumerate}
\end{theorem}

The proof is based on an analysis of the coin system
$\langle1,c_2,c_3\rangle$. If $\langle1,c_2,c_3\rangle$ is
canonical, we can decide whether $\$$ is canonical by
Theorem~\ref{th}. Otherwise, the following lemma covers the
remaining case.


\begin{lemma}\label{lem1}
Let $\$=\langle1,c_2,c_3\rangle$ be a coin system with $c_3=qc_2+r$.
If $\$$ is non-canonical, then the coin system $\$'=\langle
1,c_2,c_3,c_4\rangle$ is also non-canonical.
\end{lemma}

\begin{proof}
Since $\$$ is non-canonical, we can find the smallest counterexample
$x\in[c_3+2,c_2+c_3-1]$ by Theorem \ref{tkz1}. Assume that there is
some $c_4>c_3$ such that $\$'=\langle1,c_2,c_3,c_4\rangle$ is
canonical. We will deduce a contradiction based on the analysis of
$x$.

\begin{itemize}
\item If $x<c_4$, then $x$ is a counterexample of $\$'$, a
contradiction.

\item Otherwise, $x\geq c_4$. It is easy to see $(0,1,1,0)$ is a
representation of $c_2+c_3$, and $|\mathrm{OPT}_{\$'}(c_2+c_3)|=2$
for $c_4\leq x\leq c_2+c_3-1$. By the above assumption, we know
$\delta=c_2+c_3-c_4$ must be a coin, and here $\delta=1$. Hence,
$x=c_4=c_2+c_3-1$. By the proof of Theorem~\ref{tkz2}, we have
$x=kc_2=c_2+c_3-1$, i.e., $r=1$ and $c_3=qc_2+1$ and $c_4=qc_2+c_2$.
Thus,
\[
\$'=\langle1,c_2,qc_2+1,qc_2+c_2\rangle.
\]

For the integer $2qc_2+c_2-1$, it is easy to see that
$(c_2-3,0,2,0)$ is a representation under $\$'$, and
$\mathrm{GRD}_{\$'}(2qc_2$\\
$+c_2-1)=(c_2-1,q-1,0,1)$. Hence, we have
$|(c_2-3,0,2,0)|<|\mathrm{GRD}_{\$'}(2qc_2+c_2-1)|$, that is,
$2qc_2+c_2-1$ is a counterexample of $\$'$, a contradiction.
\end{itemize}
\end{proof}

Moreover, we prove the following Theorem~\ref{cx2} in which the coin
system with three types of coins plays a somewhat surprising role.

\begin{theorem}\label{cx2}
If a coin system $\$_1=\langle1,c_2,c_3\rangle$ is non-canonical,
then the coin system $\$_2=\langle1,c_2,c_3,\cdots,c_m\rangle$ is
also non-canonical for $m\geq4$.
\end{theorem}

Actually, we can prove the following stronger result on the
counterexamples of non-canonical systems with more than three types
of coins.

\begin{theorem}\label{cx3}
If the coin system $\$_1=\langle1,c_2,c_3\rangle$ is non-canonical,
then the coin system $\$_2=\langle1,c_2,c_3,\cdots,c_m\rangle$ is
non-canonical and there exists some counterexample $x<c_m+c_3$ for
$m\geq4$.
\end{theorem}

The proof is based on an induction on $m$ and an exhaustive
case-by-case analysis of $c_{k+1}<c_k+c_3$. The long proof is placed
in the Appendix.

\section{Coin System with Five Types of Coins}
In this section, we give a full characterization of canonical coin
systems with five types of coins.

\begin{theorem}\label{five}
A coin system $\$=\langle1,c_2,c_3,c_4,c_5\rangle$ is non-canonical
if and only if $\$$ satisfies exactly one of the following
conditions:

\begin{enumerate}
\item $\langle1,c_2,c_3\rangle$ is non-canonical.

\item $\$\neq\langle1,2,c_3,c_3+1,2c_3\rangle$.

\item $|\mathrm{GRD}_{\$}((k+1)\cdot c_4)|>k+1$ with $k\cdot c_4<c_5<(k+1)\cdot c_4$.
\end{enumerate}
\end{theorem}

The proof of this theorem is similar to that of Theorem~\ref{four}
except for the second item $\$\neq\langle1,2,c_3,c_3+1,2c_3\rangle$.
We actually need to prove the following theorem.

\begin{theorem}\label{cx1}
The coin system $\$_1=\langle1,c_2,c_3,c_4\rangle$ is non-canonical
and the coin system $\$_2=\langle1,c_2,c_3,c_4,c_5\rangle$ is
canonical if and only if $c_3>3$ and
$\$_2=\langle1,2,c_3,c_3+1,2c_3\rangle$.
\end{theorem}

The proof is based on an exhaustive case-by-case analysis of the
smallest counterexamples of some coin systems. The proof can be found
in the Appendix.

\section{Tight Coin System}
For a coin system $\$=\langle1,c_2,\cdots,c_m,c_{m+1}\rangle$, once
there is an untight subsystem $\langle1,c_2,\cdots,c_i\rangle$ with
$i\leq m+1$, $\$$ is clearly non-canonical. Therefore, it is
necessary to decide whether a tight coin system is canonical.

\begin{theorem}\label{cx4}
Let $\$_1=\langle1,c_2,c_3\rangle$,
$\$_2=\langle1,c_2,c_3,\cdots,c_m\rangle$ and
$\$_3=\langle1,c_2,c_3,\cdots,c_m,c_{m+1}\rangle$ be three tight
coin systems such that $\$_1$ is canonical but $\$_2$ is not. If
$\$_3$ is non-canonical, then there is a counterexample
$x=c_i+c_j>c_{m+1}$ of $\$_3$ with $1<c_i\leq c_j\leq c_m$.
\end{theorem}
To establish Theorem~\ref{cx4}, we first prove Lemma~\ref{lem5}.
Here, we define $c_0=0$ and $d_i:=c_i-c_{i-1}$ with $1\leq i\leq
m+1$.


\begin{lemma}\label{lem5}
Let $\$_1=\langle1,c_2,c_3\rangle$,
$\$_2=\langle1,c_2,c_3,\cdots,c_m\rangle$ and
$\$_3=\langle1,c_2,c_3,\cdots,c_m,c_{m+1}\rangle$ be three tight
coin systems. $\$_1$ is canonical but both $\$_2$ and $\$_3$ are
not. If any $c_m+c_i>c_{m+1}$ is not the counterexample of $\$_3$
with $1<c_i\leq c_m$, then $d_{m+1}=\max\{d_i\ |\ 1\leq i\leq
m+1\}$.
\end{lemma}
\begin{proof}
Assume that there is some $d_{j+1}>d_{m+1}$ with $0<j<m$. Without
loss of generality, let $d_{j+1}=d_{m+1}+\kappa$ with
$0<\kappa<d_{j+1}$.

For $c_{j+1}+c_m$, we have
$c_{j+1}+c_m=c_j+d_{j+1}+c_m=c_j+\kappa+c_{m+1}$. Since
$c_j+\kappa\in(c_j,c_{j+1})$, we have that $c_{j+1}+c_m$ is a
counterexample of $\$_3$. This is a contradiction. 
\end{proof}

\begin{proof}[Proof of Theorem~\ref{cx4}]
Assume that any $x=c_i+c_j>c_{m+1}$ with $1<c_i\leq c_j\leq c_m$ is
not the counterexample of $\$_3$. By Lemma \ref{lem5} and the
assumption, $d_{m+1}=\max\{d_i\ |\ 1\leq i\leq m+1\}$. For
simplicity, we introduce some notations.

\begin{itemize}
\item $x=c_{m+1}+\delta$ is the smallest counterexample of
$\$_3$.

\item $c_s$ is the largest coin used in all the optimal
representations of $x$ with $s\leq m$.

\item $c_l$ and $c_h$ are the smallest coin and the largest coin used
in the optimal representation of $x-c_s$ respectively.
\end{itemize}
Thus, we have $c_l>x-c_{s+1}$ by definition of $c_s$.

\begin{enumerate}[(1)]
\item If $s\leq m-1$, then $c_l>x-c_m=\delta+d_{m+1}$.

\begin{itemize}
\item If $c_s+c_l<x$, then $\mathrm{OPT}_{\$_3}(x)$ uses one coin $c_h$
besides a coin $c_l$ and a coin $c_s$. Since
$|\mathrm{GRD}_{\$_3}(x-c_h)|=|\mathrm{OPT}_{\$_3}(x-c_h)|$ and
$x-c_h\geq c_s+c_l>c_{s+1}$, we have $c_{s+1}$ appears in
$\mathrm{OPT}_{\$_3}(x)$, a contradiction.

\item Otherwise, $c_s+c_l=x>c_{m+1}$. By the assumption, there is
$c_{m+1}+c_i=c_s+c_l$, a contradiction.
\end{itemize}

\item Otherwise, $s=m$.

\begin{itemize}
\item If $c_h+c_m>c_{m+1}$, there exists $c_{m+1}+c_i=c_h+c_s$ by the
assumption. Thus, we can get a new representation of $x$ replacing
$c_h$ and $c_m$ with $c_{m+1}$ and $c_i$ in
$\mathrm{OPT}_{\$_3}(x)$. It is easy to see this new representation
uses the coin $c_{m+1}$ and remains optimality, a contradiction.

\item If $c_h+c_m<c_{m+1}$, then $\mathrm{OPT}_{\$_3}(x)$ uses
one coin $c_l$ besides a coin $c_h$ and a coin $c_m$. Since $s=m$,
we have $c_l>x-c_{m+1}$, i.e., $\delta<c_l\leq c_h<d_{m+1}$. Thus,
$c_m+\delta\in(c_m,c_{m+1})$. It is easy to see
$|\mathrm{GRD}_{\$_3}(x)|>|\mathrm{OPT}_{\$_3}(x)|=|\mathrm{GRD}_{\$_2}(x)|$.
By the assumption, there is $c_{m+1}+c_i=2c_m$. Then we have
\[
\begin{array}{rcl}
|\mathrm{GRD}_{\$_3}(c_m+\delta)|&=&|\mathrm{GRD}_{\$_3}(x)|\\
&>&1+|\mathrm{GRD}_{\$_3}(c_m+\delta-c_i)|\\
\end{array}
\]
It implies $c_m+\delta$ is a counterexample of $\$_3$, a
contradiction.

\item If $c_h+c_m=c_{m+1}$, then we can get a new representation of $x$
replacing $c_h$ and $c_m$ with $c_{m+1}$ in
$\mathrm{OPT}_{\$_3}(x)$. It is easy to see this new representation
has the smaller size, a contradiction.
\end{itemize}
\end{enumerate}
\end{proof}

\section{The Algorithm}
In this section, we present an $O(m^2)$ algorithm that decides
whether a tight coin system
$\$=\langle1,c_2,\cdots,c_m,c_{m+1}\rangle$ with $m\geq5$ is
canonical. By Theorem~\ref{cx4}, we have if $\$$ is non-canonical,
then there is a counterexample that is the sum of two coins.\medskip

\begin{center}
\parbox[c]{8cm}{
\begin{algorithm}[H]
\caption{IsCanonical}
\begin{algorithmic}[1]
\REQUIRE a tight coin system
  $\$=\langle1,c_2,\cdots,c_m,c_{m+1}\rangle$ with $m\geq5$

  \IF{$0<r<c_2-q$ with $c_3=qc_2+r$}
    \RETURN $\$$ is non-canonical
  \ELSE
    \FOR{$i=m$ downto $1$}
      \FOR{$j=i$ downto $1$}
        \IF{$c_i+c_j>c_{m+1}$ and $|\mathrm{GRD}_{\$}(c_i+c_j)|>2$}
          \RETURN $\$$ is non-canonical
        \ENDIF
      \ENDFOR
    \ENDFOR
    \RETURN $\$$ is canonical
  \ENDIF
\end{algorithmic}
\end{algorithm}
}
\end{center}
\medskip

In addition, we observe this algorithm can deal with most tight coin
systems in $2m$ steps, except a small number of non-canonical coin
systems, and they are almost arithmetic progressions, for example,
\[
\begin{array}{l}
\$=\langle1,2,\cdots,12,14,15,\cdots,20,21,24,25,26,28,29,30,39\rangle\\
\end{array}
\]
Moreover, we can characterize the smallest counterexample of such
non-canonical coin system.

\begin{lemma}\label{cx5}
Let $\$_1=\langle1,c_2,c_3\rangle$,
$\$_2=\langle1,c_2,c_3,\cdots,c_m\rangle$ and
$\$_3=\langle1,c_2,c_3,\cdots,c_m,c_{m+1}\rangle$ be three tight
coin systems such that $\$_1$ is canonical but $\$_2$ and $\$_3$ are
not. If $\$_3$ has no counterexample $x$ such that
$x=c_m+c_i>c_{m+1}$ or $x=c_{m-1}+c_j>c_{m+1}$ with $1<c_i,c_j\leq
c_m$, then the smallest counterexample of $\$_3$ is the sum of two
coins.
\end{lemma}

\begin{proof}
Modifying the proof of Theorem~\ref{cx4} slightly, it is easy to get
this proof.
\end{proof}

\section{Conclusion}
In this paper, we have given an $O(m^2)$ algorithm that decides
whether a tight coin system is canonical. Although our algorithm can
only handle tight coin systems, it is more efficient than Pearson's
algorithm. As some future work, we expect to obtain an $O(m)$
algorithm for tight canonical coin systems and an $o(m^3)$ general
algorithm.

We have shown a sufficient and necessary condition of canonical coin
systems with four or five types of coins by a novel method.
Meantime, we have also obtained a sufficient condition of
non-canonical coin systems. Many algorithm including Pearson's can
benefit from it. However, it is still left open to give full
characterizations of canonical coin systems with more than five
types of coins. It is a challenge to explore the corresponding
necessary condition in the future.

\section*{Acknowledgment}
We would like to thank Yiyuan Zheng for his valuable suggestions and
discussion.

\section*{Appendix}

\begin{proof}[Proof of Theorem~\ref{cx3}]
The proof is based on the induction on $m$ and the exhaustive
case-by-case analysis of $c_{m+1}<c_m+c_3$.

The result is trivial for $m=4$ by Lemma~\ref{lem1} and Theorem
\ref{tkz1}. Now assume that
$\$_2=\langle1,c_2,c_3,\cdots,c_k\rangle$ is non-canonical and there
exists some counterexample $x<c_k+c_3$.

Then we will check
$\$_2'=\langle1,c_2,c_3,\cdots,c_k,c_{k+1}\rangle$ based on a
detailed analysis of the various cases for $c_{k+1}$. Here, we only
consider the non-trivial cases as follows:

\begin{itemize}
\item $x>c_k$. Otherwise, $x$ is also
a counterexample of $\$_2'$ and $x<c_{k+1}+c_3$.

\item $c_i-c_{i-1}<c_3$ for $i\in[4,k+1]$.
Otherwise, there must be some integer $y$ such that $y$ is a
counterexample of $\$_2'$ and $y<c_{k+1}+c_3$ by the previous
assumption.

\item Once $c_i+c_j\in(c_s,c_{s+1})$ and $c_i+c_j<c_{k+1}+c_3$
with $i\leq j<s\leq k+1$, we have
$c_i+c_j-c_s\in\{1,c_2,c_3,\cdots,c_m\}$. Otherwise, $c_i+c_j$ is a
counterexample of $\$_2'$ and $c_i+c_j<c_{k+1}+c_3$.
\end{itemize}

Next, we will analyze $c_{k+1}<c_k+c_3$ exhaustively, and either
find a counterexample of $\$_2'$ or exclude a case for a
contradiction.

\begin{enumerate}[(1)]
\item If $c_{k+1}=c_k+1$, then $c_k+c_2$ is a
counterexample of $\$_2'$ and $c_k+c_2<c_{k+1}+c_3$.

\item If $c_{k+1}=c_k+\kappa$ with $\kappa\in(1,c_2)$,
then $c_k+c_2\in(c_{k+1},c_{k+1}+c_2-1)$ and
$c_k+c_3\in(c_{k+1},c_{k+1}+c_2-1)$. Thus, we have
\[
\begin{array}{rcl}
c_k+c_2-c_{k+1}&=&c_2-\kappa\in\{1\}\\
c_k+c_3-c_{k+1}&=&c_3-\kappa\in\{1,c_2\}\\
\end{array}
\]
It is easy to see $2c_2=c_3+1$ for $c_3-\kappa>c_2-\kappa$. Thus,
$\$_1$ is canonical by Theorem~\ref{tkz2}, a contradiction.

\item If $c_{k+1}=c_k+c_3-\kappa$ with $\kappa\in(0,c_3-c_2]$,
then $c_k+c_3\in(c_{k+1},c_{k+1}+c_3-c_2]$ and
$c_k+c_4\in(c_{k+1},c_{k+1}+c_4-c_2]$. Thus, we have
\[
\begin{array}{rcl}
c_k+c_3-c_{k+1}&=&\kappa\in\{1,c_2\}\\
c_k+c_4-c_{k+1}&=&c_4-c_3+\kappa\in\{1,c_2,c_3\}\\
\end{array}
\]

If $\kappa=c_2$, then it is easy to see $2c_2\leq c_3$ and
$c_4=2c_3-c_2\geq c_2+c_3$. By Theorem~\ref{tkz1}, the smallest
counterexample $y$ of $\$_1$ is also a counterexample of $\$_2'$ and
$y<c_{k+1}+c_3$.\medskip

If $\kappa=1$, then $c_{k+1}=c_k+c_3-1$ and
$c_4-c_3+1\in\{c_2,c_3\}$.

\begin{enumerate}
\item If $c_4-c_3+1=c_3$, then $c_4=2c_3-1$.
Since $\$_1$ is non-canonical, by the proof of Theorem \ref{tkz2},
we have the smallest counterexample $y=(q+1)c_2$ of $\$_1$ with
$0<r<c_2-q$. Replacing $c_3$ with $qc_2+r$, it is easy to see
$c_4=2qc_2+2r-1>y$. Thus, $y$ is a counterexample of $\$_2'$ and
$y<c_{k+1}+c_3$.

\item If $c_4-c_3+1=c_2$, then $c_4=(q+1)c_2+r-1$ with $c_3=qc_2+r$.
Similarly, we have the smallest counterexample $y=(q+1)c_2$ of
$\$_1$ with $0<r<c_2-q$.\medskip

If $r>1$, then $y$ is a counterexample of $\$_2'$ and
$y<c_{k+1}+c_3$.\smallskip

If $r=1$, then $c_3=qc_2+1$ and $c_4=(q+1)c_2$. Since
$c_{k-1}>c_k-c_3$ and $c_{k+1}-c_k=c_3-1$, we have
$c_{k-1}+c_4\in(c_k+c_2-1,c_k+c_4)$ and $c_k+c_4>c_{k+1}$.\smallskip

\begin{itemize}
\item If $c_{k-1}+c_4\in(c_k+c_2-1,c_{k+1})$, then
$c_{k-1}+c_4-c_k\in\{c_2\}$, i.e.,
$c_k-c_{k-1}=c_{k+1}-c_k=c_4-c_2=qc_2$. Since $c_{k-1}<c_{k-2}+c_3$,
we have $c_{k-2}+c_3\in(c_{k-1},c_k]$.\medskip

If $c_{k-2}+c_3=c_k$, then $c_{k-1}-c_{k-2}=1$. It is easy to see
$c_{k-2}+c_2\in(c_{k-1},c_k)$ is a counterexample of
$\$_2'$.\smallskip

If $c_{k-2}+c_3<c_k$, then $c_{k-2}+c_3-c_{k-1}\in\{1,c_2\}$.

\begin{itemize}
\item If $c_{k-2}+c_3-c_{k-1}=c_2$, then
$c_{k-1}-c_{k-2}=c_3-c_2$.\medskip

If $c_3-c_2>c_2$, then $2c_3=c_4+c_2$. Since $c_4=(q+1)c_2$, we have
$c_2=\frac{2}{2-q}$ and $q\geq1$, i.e., $c_2=2$. However,
$\$_1=\langle1,2,3\rangle$ is canonical, a contradiction.\smallskip

If $c_3-c_2<c_2$, then $q=1$ and $c_3=c_2+1$ and $c_4=2c_2$. Thus,
$c_{k-2}+c_2\in(c_{k-1},c_k)$ is a counterexample of
$\$_2'$.\medskip

\item If $c_{k-2}+c_3-c_{k-1}=1$, then
$c_{k-1}-c_{k-2}=c_3-1$.

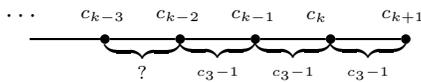
\begin{pspicture}(1,-0.3)(13,1.3)
\uput[0](1,0.8){$\cdots$}

\uput[0](2,0.8){$\scriptstyle c_{k-3}$}

\uput[0](3,0.8){$\scriptstyle c_{k-2}$}

\uput[0](4,0.8){$\scriptstyle c_{k-1}$}

\uput[0](5,0.8){$\scriptstyle c_k$}

\uput[0](6,0.8){$\scriptstyle c_{k+1}$}

\pnode(1.5,0.5){O}

\dotnode(2.5,0.5){A}

\dotnode(3.5,0.5){B}

\dotnode(4.5,0.5){C}

\dotnode(5.5,0.5){D}

\dotnode(6.5,0.5){E}

\ncline{O}{A}

\ncline{A}{B}

\ncline{B}{C}

\ncline{C}{D}

\ncline{D}{E}

\uput[0](2.32,0.2){$\underbrace{\hspace{1cm}}_{\scriptstyle ?}$}

\uput[0](3.32,0.2){$\underbrace{\hspace{1cm}}_{\scriptscriptstyle
c_3-1}$}

\uput[0](4.32,0.2){$\underbrace{\hspace{1cm}}_{\scriptscriptstyle
c_3-1}$}

\uput[0](5.32,0.2){$\underbrace{\hspace{1cm}}_{\scriptscriptstyle
c_3-1}$}

\end{pspicture}

If $c_{k-2}-c_{k-3}<c_3-1$, then there is a counterexample $z$ of
$\langle1,c_2,\cdots,c_{k-3}\rangle$ such that $z<c_{k-3}+c_3\leq
c_{k-1}$ by the previous assumption.\smallskip

\begin{itemize}
\item If $z\!<\!c_{k-2}$, then $z\!<\!c_{k+1}\!+c_3$ is a counterexample
of~$\$_2'$.

\item If $c_{k-2}\leq z<c_{k-3}+c_3-1$, then
$c_{k-2}+z-c_{k-3}\in(c_{k-2},c_{k-1})$, which is a counterexample
of $\$_2'$.

\item If $z=c_{k-3}+c_3-1$, then $c_{k+1}+c_3-1$ is a counterexample
of $\$_2'$.\smallskip
\end{itemize}

Otherwise, $c_{k-2}-c_{k-3}=c_3-1$. Similarly, we either find a
counterexample of $\$_2'$ or obtain
\[
c_{k-3}-c_{k-4}=\cdots=c_4-c_3=c_3-1=qc_2.
\]
Thus, $c_2=\frac{1}{1-q}$, a contradiction.\medskip
\end{itemize}

\item If $c_{k-1}+c_4=c_{k+1}$, then $c_k-c_{k-1}=c_2$. Since
$c_{k-1}+c_3=c_k+c_3-c_2\in(c_k,c_{k+1})$, we have
$c_3-c_2\in\{1,c_2\}$.

\begin{itemize}
\item If $c_3-c_2=c_2$, then
$c_3=2c_2$, a contradiction.

\item If $c_3-c_2=1$, then $c_3=c_2+1$ and $c_4=2c_2$. Since
$c_{k-2}+c_3\in(c_{k-1},c_k)$, we have
$c_{k-2}+c_3-c_{k-1}\in\{1\}$, i.e., $c_{k-1}-c_{k-2}=c_2$. Similar
to the above proof, we either find a counterexample of $\$_2'$ or
obtain
\[
c_{k-2}-c_{k-3}=\cdots=c_4-c_3=c_2.
\]
It is a contradiction.\medskip
\end{itemize}

\item If $c_{k-1}+c_4\in(c_{k+1},c_{k+1}+c_2)$, then
$c_{k-1}+c_4-c_{k+1}\in\{1\}$. Since $c_{k+1}=c_k+c_3-1$, we have
$c_k-c_{k-1}=c_2-1$, and $c_{k-1}+c_3=c_k+c_3+1-c_2$. Since
$c_3+1-c_2\in\{c_2\}$, we have $c_3=2c_2-1$, a contradiction.
\end{itemize}
\end{enumerate}
\end{enumerate}
\end{proof}
\medskip

\begin{proof}[Proof of Theorem~\ref{cx1}]
To show the above theorem, we first need to prove the following
lemma.

\begin{lemma}\label{lem3}
Let $\$_1=\langle1,c_2,c_3,c_2+c_3-1\rangle$ and
$\$_2=\langle1,c_2,c_3,c_2+c_3-1,c_2+2c_3-2\rangle$ be two coin
systems such that $\$_2$ is canonical but $\$_1$ is not. Then the
coin system $\$_3=\langle1,c_2,c_3\rangle$ is canonical.
\end{lemma}

\begin{proof}
Assume that $\$_3$ is non-canonical. By Theorem \ref{tkz1}, there
exists the smallest counterexample $y\in[c_3+2,c_2+c_3-1]$ of
$\$_3$. Next, we will deduce a contradiction by analyzing $y$ in
detail.

\begin{enumerate}[(1)]
\item If $y\in[c_3+2,c_2+c_3-2]$, then $y$ is also a counterexample of
$\$_2$, a contradiction.

\item Otherwise, $y=c_2+c_3-1=kc_2$ with $1<k<c_2$ by the proof of
Theorem~\ref{tkz2}. Thus, we have
\[
\$_2=\langle1,c_2,kc_2-c_2+1,kc_2,2kc_2-c_2\rangle
\]

Since $\$_1$ is non-canonical, we have the smallest counterexample
$x\in[c_3+2,c_2+2c_3-2]$ by Theorem \ref{tkz1}.\smallskip

\begin{itemize}
\item If $x\in[c_3+2,c_2+2c_3-3]$, then $x$ is also a counterexample of
$\$_2$, a contradiction.

\item Otherwise, $x=c_2+2c_3-2=2kc_2-c_2$. By Theorem \ref{dis},
we have
\[
\begin{array}{rcll}
\mathrm{GRD}_{\$_1}(x)&=&(0,k-1,0,1)\\
\mathrm{OPT}_{\$_1}(x)&=&(\beta_1,0,\beta_3,0)&\mathrm{with}\
\beta_3\leq2\\
\end{array}
\]

\begin{itemize}
\item If $\beta_3\leq1$, then $\beta_1>k-1$. It is easy to see
$|\mathrm{OPT}_{\$_1}(x)|\geq|\mathrm{GRD}_{\$_1}(x)|$, a
contradiction.

\item If $\beta_3=2$, then $\beta_1=c_2-2>k-2$. Thus, we have
$|\mathrm{OPT}_{\$_1}(x)|>|\mathrm{GRD}_{\$_1}(x)|$, a
contradiction.
\end{itemize}
\end{itemize}
\end{enumerate}
Therefore, $\$_3=\langle1,c_2,c_3\rangle$ is canonical. 
\end{proof}

\begin{enumerate}
\item[($\Leftarrow$)]

First, we show that $\$_1=\langle1,2,c_3,c_3+1\rangle$ is
non-canonical with $c_3>3$. It is easy to see
$\mathrm{OPT}_{\$_1}(2c_3)=(0,0,2,0)$. By Theorem \ref{dis}, we have
$\mathrm{GRD}_{\$_1}(2c_3)=(\alpha_1,\alpha_2,0,1)$. Since $c_3>3$,
we have $|\mathrm{GRD}_{\$_1}(c_3-1)|>1$, i.e.,
$|\mathrm{GRD}_{\$_1}(2c_3)|>2$. Thus, $\$_1$ is non-canonical.
\medskip\smallskip

Secondly, we show that $\$_2=\langle1,2,c_3,c_3+1,2c_3\rangle$ is
canonical. Assume that $\$_2$ is non-canonical. By Theorem
\ref{tkz1}, there is the smallest counterexample $x$ of $\$_2$ such
that $x\in[c_3+2,3c_3)$ and $x\neq2c_3$. Next, we will deduce a
contradiction by analyzing $x$ in detail.\medskip

\begin{enumerate}[(1)]
\item If $x\in[c_3+2,2c_3)$, then it is also the smallest counterexample of
$\$_1$. For simplicity, let $x=c_3+\kappa$ with $\kappa\in[2,c_3)$.

\begin{itemize}
\item If $\kappa=2\ell+1$, then
$1\leq \ell\leq\frac{c_3-1}{2}-1$. By Theorem~\ref{dis}, we have
\[
\begin{array}{rcl}
\mathrm{GRD}_{\$_1}(x)&=&(0,\ell,0,1)\\
\mathrm{OPT}_{\$_1}(x)&=&(\beta_1,\beta_2,1,0)\\
\end{array}
\]

Thus, $\beta_1+2\beta_2=2\ell+1$ and $\beta_1\in\{0,1\}$. It is easy
to see $|\mathrm{GRD}_{\$_1}(x)|<|\mathrm{OPT}_{\$_1}(x)|$, a
contradiction.\smallskip

\item If $\kappa=2\ell$, then $2\leq \ell\leq\frac{c_3}{2}-1$.
Similarly, we have
\[
\begin{array}{rcl}
\mathrm{GRD}_{\$_1}(x)&=&(1,\ell-1,0,1)\\
\mathrm{OPT}_{\$_1}(x)&=&(0,\ell,1,0)\\
\end{array}
\]
Obviously, $|\mathrm{OPT}_{\$_1}(x)|=|\mathrm{GRD}_{\$_1}(x)|$, a
contradiction.\medskip
\end{itemize}

\item Otherwise, $x\in(2c_3,3c_3)$. Let
$x=2c_3+\kappa$ with $\kappa\in[2,c_3)$.

\begin{itemize}
\item If $\kappa=2\ell+1$, then
$1\leq \ell\leq\frac{c_3-1}{2}-1$. By Theorem \ref{dis}, we have
\[
\begin{array}{rcl}
\mathrm{GRD}_{\$_2}(x)&=&(1,\ell,0,0,1)\\
\mathrm{OPT}_{\$_2}(x)&=&(0,\beta_2,\beta_3,\beta_4,0)\\
\end{array}
\]

If $\beta_3+\beta_4\geq3$, then
$2\beta_2+\beta_3c_3+(c_3+1)\beta_4\geq3c_3>x$, a
contradiction.\smallskip

And for $\beta_3+\beta_4=0$, $\beta_3+\beta_4=1$ and
$\beta_3+\beta_4=2$, it is easy to get a contradiction
similarly.\smallskip

\item If $\kappa=2\ell$, then
we can also deduce a contradiction similar to the above
case.\smallskip
\end{itemize}
\end{enumerate}
Thus, $\$_2=\langle1,2,c_3,c_3+1,2c_3\rangle$ is
canonical.\medskip\smallskip

\item[($\Rightarrow$)]
For the integer $c_3+c_4$, it is easy to see
$\mathrm{GRD}_{\$_2}(c_3+c_4)=(\alpha_1,\alpha_2,0,0,1)$ and
$\mathrm{OPT}_{\$_2}(c_3+c_4)=(0,0,1,1,0)$. Since $\$_2$ is
canonical, we have either $\alpha_1=1,\alpha_2=0$ or
$\alpha_1=0,\alpha_2=1$, i.e., {
\renewcommand{\theequation}{A}
\begin{equation}\label{eq1}
c_5=c_3+c_4-1\text{ or }c_5=c_3+c_4-c_2.
\end{equation}
} For the integer $2c_4$, we have {
\renewcommand{\theequation}{B}
\begin{equation}\label{eq2}
c_5=2c_4-1\text{ or }c_5=2c_4-c_2\text{ or }c_5=2c_4-c_3
\end{equation}
} Correlating (\ref{eq1}) with (\ref{eq2}), we have 3 feasible
equations as follows:
\[
\begin{array}{l}
\textcircled{\scriptsize 1}\left\{\begin{array}{lcl}c_5&=&c_3+c_4-c_2\\
c_5&=&2c_4-c_3\\
\end{array}\right.\textcircled{\scriptsize 2}\left\{\begin{array}{lcl}c_5&=&c_3+c_4-1\\
c_5&=&2c_4-c_3\\
\end{array}\right.
\end{array}
\]
\[
\textcircled{\scriptsize 3}\left\{
\begin{array}{lcl}c_5&=&c_3+c_4-1\\
c_5&=&2c_4-c_2\\
\end{array}\right.
\]

Next, we will deduce a contradiction from \textcircled{\scriptsize
1} and \textcircled{\scriptsize 2} respectively.

\begin{enumerate}[(1)]
\item Solving \textcircled{\scriptsize 1}, we have $c_4=2c_3-c_2$ and
$c_5=3c_3-2c_2$. Thus,
\[
\$_2=\langle1,c_2,c_3,2c_3-c_2,3c_3-2c_2\rangle
\]
Since $\$_1$ is non-canonical, there is the smallest counterexample
$x\in[c_3+2,3c_3-c_2-1]$ of $\$_1$.\medskip

If $x\in[c_3+2,3c_3-2c_2-1]$, then $x$ is also a counterexample of
$\$_2$, contradiction.\smallskip

Otherwise, $x\in[3c_3-2c_2,3c_3-c_2-1]$. By Theorem~\ref{dis}, we
have
\[
\begin{array}{rcll}
\mathrm{GRD}_{\$_1}(x)&=&(\alpha_1,\alpha_2,0,1)\\
\mathrm{OPT}_{\$_1}(x)&=&(\beta_1,\beta_2,\beta_3,0)&\mathrm{with}\
\beta_3\leq2
\end{array}
\]

\begin{itemize}
\item If $c_3\geq2c_2$, then
$\mathrm{GRD}_{\$_1}(x)=(\alpha_1,\alpha_2,0,1)$ and
$\mathrm{OPT}_{\$_1}(x)=(\beta_1,0,\beta_3,0)$ where
$\beta_3\in\{0,1,2\}$ and $\alpha_2>0$. Since
$x\in[3c_3-2c_2,3c_3-c_2-1]$, we have $|\mathrm{GRD}_{\$_1}(x)|\leq
c_3-2c_2+2$.\medskip

For $\beta_3=0$, $\beta_3=1$ and $\beta_3=2$, it is easy to deduce a
contradiction respectively.\smallskip

\item If $c_3<2c_2$, then $\mathrm{GRD}_{\$_1}(x)=(\alpha_1,0,0,1)$ and
$\mathrm{OPT}_{\$_1}(x)=(0,\beta_2,\beta_3,0)$ where
$\beta_3\in\{0,1,2\}$ and $\beta_2+\beta_3<1+\alpha_1$.

\begin{itemize}
\item If $\beta_3=2$, then $\beta_2=0$ and $\alpha_1=c_2$, a
contradiction.\smallskip

\item If $\beta_3=1$, then $c_3=(\beta_2+1)c_2-\alpha_1$. Since
$c_3\in(c_2,2c_2)$, we have $c_3=2c_2-\alpha_1$. Thus,
\[
\$_2=\langle1,c_2,2c_2-\alpha_1,3c_2-2\alpha_1,4c_2-3\alpha_1\rangle
\]

If $\alpha_1=1$, then
$\$_2=\langle1,c_2,2c_2-1,3c_2-2,4c_2-3\rangle$. It is easy to see
$\langle1,c_2,2c_2-1,3c_2-2\rangle$ is canonical, a
contradiction.\smallskip

If $\alpha_1>1$, then $4c_2-2\alpha_1$ is a counterexample of
$\$_2$, a contradiction.\smallskip

\item If $\beta_3=0$, then
$c_3=\frac{\beta_2-1}{2}c_2+(c_2-\frac{\alpha_1}{2})=qc_2+r$. Since
$c_3<2c_2$, we have $\beta_2=3$ and $x=3c_2$ and
$q=1,r=c_2-\frac{1}{2}\alpha_1$. Thus,
\[
\$_2=\langle1,c_2,c_2+r,c_2+2r,c_2+3r\rangle
\]
Since $\beta_2+\beta_3<1+\alpha_1$, we have $\alpha_1>2$. By
Theorem~\ref{tkz2}, $\$_3=\langle1,c_2,c_2+r\rangle$ is
non-canonical. It is easy to see $2c_2<c_2+2r$ is the smallest
counterexample of $\$_2$, a contradiction.\medskip

\end{itemize}

\end{itemize}

\item Solving \textcircled{\scriptsize 2}, we have $c_4=2c_3-1$ and
$c_5=3c_3-2$. Thus,
\[
\$_2=\langle1,c_2,c_3,2c_3-2,3c_3-2\rangle
\]
It is easy to see $c_3+c_4-1<c_5$, i.e., the smallest counterexample
of $\$_1$ is also a counterexample of $\$_2$, a
contradiction.\medskip

\item Solving \textcircled{\scriptsize 3}, we have $c_4=c_3+c_2-1$ and
$c_5=2c_3+c_2-2$. Thus,
\[
\$_2=\langle1,c_2,c_3,c_2+c_3-1,c_2+2c_3-2\rangle
\]
By Lemma \ref{lem3}, $\langle1,c_2,c_3\rangle$ is canonical. Since
$\$_1$ is non-canonical, $2c_3$ is a counterexample of $\$_1$ by
Theorem \ref{th}. And we claim $2c_3\geq c_2+2c_3-2$. Otherwise,
$2c_3$ is a counterexample of $\$_2$, a contradiction. Thus,
$c_2\leq2$, i.e.,
\[
\$_2=\langle1,2,c_3,c_3+1,2c_3\rangle
\]
\end{enumerate}
\end{enumerate}
This completes the proof.
\end{proof}

\medskip

\begin{proof}[Proof of Lemma~\ref{cx5}]
First, we introduce some notation as follows:

\begin{itemize}
\item $x=c_{m+1}+\delta$ is the smallest counterexample of
$\$_3$.

\item $c_s$ is the largest coin used in all the optimal
representations of $x$ with $s\leq m$.

\item $c_l$ and $c_h$ are the smallest coin and the largest coin used
in the optimal representation of $x-c_s$ respectively.
\end{itemize}

Thus, we have $c_l>x-c_{s+1}$ by definition of $c_s$. Since $x$ is
the smallest counterexample of $\$_3$, we have
\[
\begin{array}{rcl}
|\mathrm{GRD}_{\$_3}(x-c_s)|&=&|\mathrm{OPT}_{\$_3}(x-c_s)|\\
|\mathrm{GRD}_{\$_3}(x-c_l)|&=&|\mathrm{OPT}_{\$_3}(x-c_l)|\\
|\mathrm{GRD}_{\$_3}(x-c_h)|&=&|\mathrm{OPT}_{\$_3}(x-c_h)|\\
\end{array}
\]

Assume that $x$ is not the sum of two coins, i.e., $c_s+c_l<x$.
Thus, we can find one coin $c_h$ besides a coin $c_l$ and a coin
$c_s$ in the optimal representation of $x$.

\begin{enumerate}[(1)]
\item If $s\leq m-1$, then $c_l>x-c_m=\delta+d_{m+1}$. By
Lemma~\ref{lem5}, it is easy to see $x-c_h\geq c_s+c_l>c_{s+1}$.
Thus, $c_{s+1}$ should appear in the optimal representation of $x$,
a contradiction.\smallskip

\item Otherwise, $s=m$.

\begin{itemize}
\item If $c_h+c_m>c_{m+1}$, there exists $c_{m+1}+c_i=c_h+c_s$ by the
hypothesis. We can replace $c_h$ and $c_m$ with $c_{m+1}$ and $c_i$,
a contradiction.

\item If $c_h+c_m<c_{m+1}$, then $\mathrm{OPT}_{\$_3}(x)$ uses one coin
$c_l$ besides a coin $c_h$ and a coin $c_m$. Since $s=m$, we have
$c_l>x-c_{m+1}$, i.e., $\delta<c_l\leq c_h<d_{m+1}$. By the
hypothesis, we have $c_{m+1}+c_i=2c_m$. Consider
$c_m+\delta<c_{m+1}$.
\[
\begin{array}{rcl}
|\mathrm{GRD}_{\$_3}(c_m+\delta)|&=&1+|\mathrm{GRD}_{\$_3}(\delta)|\\
&>&1+|\mathrm{GRD}_{\$_3}(\delta+d_{m+1})|\\
&=&1+|\mathrm{GRD}_{\$_3}(c_m+\delta-c_i)|\\
\end{array}
\]
Therefore, $c_m+\delta$ is a counterexample of $\$_3$, a
contradiction.

\item If $c_h+c_m=c_{m+1}$, then $\mathrm{OPT}_{\$_3}(x)$ uses
one coin $c_l$ besides a coin $c_h$ and a coin $c_m$. Thus,
$\delta<c_l\leq c_h=d_{m+1}$.
\[
\begin{array}{rcl}
|\mathrm{GRD}_{\$_3}(x)|&=&1+|\mathrm{GRD}_{\$_3}(\delta)|\\
&<&2+|\mathrm{GRD}_{\$_3}(\delta)|=|\mathrm{OPT}_{\$_3}(x)|
\end{array}
\]
This is a contraction.

\end{itemize}
\end{enumerate}
\end{proof}

\end{document}